\documentclass{article}
\usepackage{amssymb}
\renewcommand{\leq}{\leqslant}
\renewcommand{\geq}{\geqslant}
\usepackage{amsmath}
\usepackage{amsthm}
\usepackage{tikz}
\newtheorem{theorem}{Theorem}
\newtheorem{proposition}[theorem]{Proposition}
\theoremstyle{remark}
\newtheorem{example}[theorem]{Example}
\newtheorem{remark}[theorem]{Remark}
\newcommand{\dat}{\text{DA}_2}
\newcommand{\dato}{\text{DA}^1_2}
\newcommand{\edge}[1]{\overset{#1}\to}
\newcommand{\adom}[1]{\mathrm{adom}(#1)}
\begin{document}
\title{Undecidability of satisfiability in
the algebra of finite binary relations with
union, composition, and difference}
\author{Tony Tan, Jan Van den Bussche, and Xiaowang Zhang \\
Hasselt University}
\maketitle
\begin{abstract}

We consider expressions built up from binary relation names using
the operators union, composition, and set difference.  We show
that it is undecidable to test whether a given such expression
$e$ is finitely satisfiable, i.e., whether there exist finite
binary relations that can be substituted for the relation names
so that $e$ evaluates to a nonempty result.  This result
already holds in restriction to expressions that mention just a
single relation name, and where the difference operator can be
nested at most once.

\end{abstract}

\section{Introduction}

The calculus (or algebra) of binary relations was invented by Peirce and
Schr\"oder and further developed by Tarski and his collaborators
\cite{tarski_relcalc,pratt_relcalc,maddux_originra}.  Hence we
will denote it by TA (for Tarski Algebra).
TA consists
of the operators union, complement, composition, and inverse, and
provides the empty and the identity relations as constants.  At
present, this algebra (often extended with the transitive
closure operator) provides a nice theoretical foundation for
query languages for graph databases modeled as finite binary
relational structures
\cite{good_tarski,catemarx_sigmodrecord,rafragments,wood_survey,gxpath}.
Also practical graph database query languages such as Gremlin fit
in this framework.

Specifically, given a vocabulary $\Gamma$ of binary relation
names, we can consider expressions built up using the names in
$\Gamma$ and the constants and operators mentioned above.  These
expressions serve as abstractions of query expressions evaluated
on graph databases, viewed as relational structures over
$\Gamma$.  The result of a query is again a binary relation.  For
example, for $a,b \in \Gamma$, the expression $aaa - b$
asks for all pairs $(x,y)$ such that one can walk from
$x$ to $y$ in three steps using $a$-edges, but there is no direct
$b$-edge from $x$ to $y$. Here, the operation of composition is
denoted simply by juxtaposition, and $-$
(set difference) can be expressed in terms of union and
complement by $r - s = (r^c \cup s)^c$.

In this manner, one can express precisely the binary queries
definable in FO$^3$, the fragment of first-order logic with three
variables \cite{tarskigivant,marxvenema_multi}.  In particular,
one can translate effectively from an FO$^3$ formula
with two free variables to a TA expression,
and back.  This connection with first-order logic provides
immediate insight in the classical decision problem in the context of TA:
given a vocabulary $\Gamma$ and a TA expression $e$ over
$\Gamma$, is $e$ satisfiable?  That is, does there exists a
structure $I$ over $\Gamma$ such that $e$ on $I$ evaluates to a
nonempty result?  Since satisfiability for FO$^3$ is undecidable
\cite{bgg_decision,schoenfeld_undecfinra},
satisfiability for TA is undecidable as well.

This undecidability result can be sharpened considerably: it
already holds for the fragment of TA consisting only of union,
complement, and composition \cite{andreka_memoir}.  In this
paper, we show that undecidability continues to hold when we have
only the relative form of complement provided by the set
difference operation.  Concretely, we consider a fragment of TA
that we call the Downward Algebra (DA): its only operators are
union, intersection, composition, and set difference.  The name
of this fragment is inspired by its salient property that, when
viewing binary relations as directed graphs, DA expressions can
only talk about pairs of elements formed by following edges in
the forward (or downward) direction.\footnote{A similar
terminology has been used in the context of XPath, which is a
form of TA used on tree structures as opposed to general graphs
\cite{figueira_downward}.} The focus on set difference, as
opposed to general complement, is motivated by the database query
language setting, where set difference is the standard form of
negation \cite{ahv_book}.  We will actually show that
undecidability already holds for DA expressions in which the
nesting depth of difference operators is at most two, and that
use only a single relation name.  We denote this fragment of DA
by $\dato$.

Our result is also relevant to expressive description logics and
dynamic logics.  Indeed, DA expressions can be viewed as extended
`role' expressions in description logic, or `programs' in dynamic
logic \cite{dlhandbook,dynamiclogicbook}, so our result shows
that satisfiability of such extended expressions or formulas is
undecidable already for $\dato$.  Known undecidability results
for expressive dynamic/description logics assume either the full
complement or the transitive closure operator
\cite{vrgoc_containment}.  An undecidability proof given by Lutz
and Walther \cite{lutz_atomneg} also uses only set difference on
binary relations, but additionally needs the identity relation
and the `diamond' operator $\langle r \rangle = \{(x,x) \mid \exists y :
(x,y) \in r\}$ on binary relations.  On the other hand, dynamic
logic where complement can be applied only relative to the
identity relation (so-called `formula negation'), as well as to
relation names (so-called `atomic negation'), is still decidable
\cite{lutz_atomneg}.  Thus, our result sharpens known
undecidability results and helps delineating the boundary of
undecidability.  We repeat that DA contains neither the identity
relation nor the diamond operator.

We should make clear that our result is specifically about
satisfiability by a \emph{finite} structure.  The problem of
deciding unrestricted satisfiability for DA expressions remains
open.  

This paper is further organized as follows.
Section~\ref{secdefin} defines DA, the fragment $\dato$, and the
corresponding satisfiability problem formally.
Section~\ref{secproof} proves undecidability of finite
satisfiability for general vocabularies.  Section~\ref{secone}
reduces the problem to a vocabulary with just a single relation
name.  Section~\ref{seconcl} concludes.

\section{Satisfiability of DA expressions} \label{secdefin}

Let $\Gamma$ denote a finite vocabulary of binary relation names.  The
expressions $e$ of DA over $\Gamma$ are defined by the following
grammar, where $a$ ranges over the elements of $\Gamma$:
$$ e ::= a \mid e \cup e \mid e \cap e \mid e - e \mid e \cdot e $$
The dot operator, which will denote composition, is often omitted
when writing expressions,
thus denoting composition simply by juxtaposition.  For
example, for $a,b \in \Gamma$, the expression $(a\cdot a -b)
\cdot a$ is also written as $(aa-b)a$.

A \emph{structure} over $\Gamma$ is a mapping $I$ assigning to
every $a \in \Gamma$ a binary relation $a^I$.  In this
paper, we focus on finite structures, so the binary relations
$a^I$ must be finite unless explicitly specified otherwise.
It is natural to view such a structure as a directed graph where
edges are labeled by relation names.  Accordingly we will refer
to a pair $(x,y)$ in $a^I$ as an `$a$-edge' and denote it by $x \edge a
y$.

The relation
defined by an expression $e$ in a structure $I$, denoted by
$e(I)$, is defined inductively as follows:
\begin{itemize}
\item
$a(I) = a^I$;
\item
$(e_1 \cup e_2)(I) = e_1(I) \cup e_2(I)$;
\item
$(e_1 \cap e_2)(I) = e_1(I) \cap e_2(I)$;
\item
$(e_1 - e_2)(I) = \{(x,y) \in e_1(I) \mid (x,y) \notin e_2(I)\}$;
\item
$(e_1 \cdot e_2)(I) = \{(x,y) \mid \exists z : (x,z) \in e_1(I)$
and $(z,y) \in e_2(I)\}$.
\end{itemize}

An expression $e$ over $\Gamma$ is called \emph{finitely
satisfiable} if there exists a structure $I$ over $\Gamma$ such
that $e(I)$ is nonempty.

\begin{remark}
The standard notion of structure would include an explicit set $U$,
called the domain of the structure, so that the relations $a^I$
are binary relations on $U$.  In the presence of a complementation
operation this is important, as then the complement of a relation
in a structure with domain $U$ is taken with respect to $U \times
U$.  In our setting, however, we only have set difference, so an
explicit domain would be irrelevant.  Our notion of structure
without an explicit domain actually agrees with the standard
notion of `database instance' in database theory \cite{ahv_book}.
\qed
\end{remark}

\begin{example}
A trivial example of an unsatisfiable expression is
$a-a$, but here is a less trivial example.  For relation names
$a$ and $b$, the expression $$ aaa - ((aa-b)a \cup ba) $$ is
neither finitely satisfiable nor satisfiable by an infinite
structure.  In proof, consider a pair $(x,y)$ that would belong
to the result of evaluating this expression in some structure
(for brevity we are omitting explicit reference to this
structure).  Then $(x,y) \in aaa$ so there exist $a$-edges
$(x,x_1)$, $(x_1,x_2)$, and $(x_2,y)$.  Since $(x,y) \notin
(aa-b)a$, the $b$-edge $(x,x_2)$ must be present.  But then
$(x,y) \in ba$, which is in contradiction with the last part of
the expression.
\qed
\end{example}

\begin{example}
Expressions not involving the difference operator are always
satisfiable, even by a finite series-parallel graph
\cite{dg_nfbr}.  Using difference, we can give an expression $e$
that is finitely satisfiable, but not by a series-parallel graph:
$$ a(a \cap aa) - (aa-a)a $$
Indeed we have $(1,4) \in e(W)$ where $W$ is the canonical
non-series-parallel graph \cite{valdes_seriesparallel}:
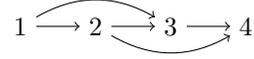
$$
\begin{tikzpicture}
\node (1) at (0,0) {$1$};
\node (2) at (1,0) {$2$};
\node (3) at (2,0) {$3$};
\node (4) at (3,0) {$4$};
\path (1) edge[->] (2)
      (2) edge[->] (3)
      (3) edge[->] (4);
\path (1) edge[out=30,in=150,->] (3);
\path (2) edge[out=-30,in=-150,->] (4);
\end{tikzpicture}
$$

To see that $e$ cannot be satisfied by any series-parallel graph,
suppose $(x,y)$ belongs to the result of evaluating $e$ on some
structure.
Since $(x,y) \in a(a\cap aa)$, there
exist edges $x \to u_1 \to u_2 \to y$ and $u_1 \to y$
(we omit the labels on the edges which are all $a$).
Since $(x,y) \notin (aa-a)a$, there must be an edge $x \to u_2$.
If at least two of the four elements $x$, $u_1$, $u_2$ and $y$
are identical, the graph contains a cycle and is not
series-parallel.  If all four elements are distinct,
we have a subgraph isomorphic to $W$ above, so the
structure is not series-parallel \cite{valdes_seriesparallel}.
\qed
\end{example}

\begin{example} \label{infinity}
We can also give an example of an `infinity axiom' in DA:
an expression that is not
finitely satisfiable but that is infinitely satisfiable.  Let $c$
be a third relation name apart from $a$ and $b$, and consider the
following expression $e$:
$$ aba - \bigl ( a(ba-a) \cup (a-ab)a \cup a(b-c)a \cup a(cc-c)a
\cup a(cb\cap b)a \bigr ) $$
To see that $e$ is infinitely satisfiable, denote the set of
natural numbers without zero by $\mathbf N$.  Let $\infty$ denote
an element that is neither zero nor in $\mathbf N$.  Now consider the
infinite structure $I$ where
\begin{align*}
a^I & = \{(0,i) \mid i \in \mathbf N \ \& \ i \geq 2\} \cup
\{(i,\infty) \mid i \in \mathbf N\} \\
b^I & = \{(i+1,i) \mid i \in \mathbf N\} \\
c^I & = \{(j,i) \mid i,j \in \mathbf N \ \& \ j>i\}
\end{align*}
Then one can verify that $(0,\infty) \in e(I)$.

To see that $e$ is not finitely satisfiable,
suppose that $(x,y)$ would belong to the relation defined by $e$ in some finite
structure.  Then $(x,y) \in aba$ so there exist edges
$x \edge a u_2 \edge b u_1 \edge a y$.
Since $(x,y) \notin a(ba-a)$ we have also $u_2 \edge a y$.

Since $(x,y) \notin (a-ab)a$, there must exist edges $x \edge a
u_3 \edge b u_2$.  Again since $(x,y) \notin a(ba-a)$ we have
also $u_3 \edge a y$.  Continuing in this fashion we obtain an
infinite sequence $u_1,u_2,\dots$ with edges $x \edge a u_i$ for
every $i \geq 2$ and edges $u_i \edge a y$ and $u_{i+1} \edge b
u_i$ for every $i \geq 1$.

Now since $(x,y) \notin a(b-c)a$ we have $u_{i+1} \edge c u_i$ for
every $i \geq 1$.  Then since $(x,y) \notin a(cc-c)a$ we have
$u_j \edge c u_i$ for all $j>i\geq 1$.
Since the structure is finite, there must exist $1 \leq i < j$ so that
$u_i=u_j$.  Hence we have a self-loop $u_j \edge c u_j$, implying
$(x,y) \in a(cb\cap b)a$ which is in contradiction with the
last part of the expression $e$.
\qed
\end{example}

The \emph{finite satisfiability problem}
for DA takes as input $\Gamma$ and $e$, and asks to
decide whether $e$ is finitely satisfiable.  We will show that this
problem is undecidable already when $\Gamma$ consists of a single
relation name, and the difference degree of $e$ is at most two.
Here, the \emph{difference degree}, denoted by $\deg e$,
indicates how deeply applications of the difference operator are
nested, and is inductively defined as follows:
\begin{itemize}
\item
$\deg a = 0$;
\item
$\deg (e_1 \cup e_2) = \deg (e_1 \cap e_2) = \deg (e_1 \cdot e_2)
= \max (\deg e_1, \deg e_2)$;
\item
$\deg (e_1 - e_2) = \max (\deg e_1, \deg e_2) + 1$.
\end{itemize}
The set of expressions with difference degree at most two is
denoted by $\dat$.  The set of $\dat$ expressions over a
single relation name is denoted by $\dato$.  In
Section~\ref{secproof}, we will show that finite satisfiability for
$\dat$ is undecidable; in Section~\ref{secone} we will show that
this already holds for $\dato$.

\begin{remark} \label{remarkintersection}
Our focus on $\dat$ explains why we have included intersection in
DA, while this operator is actually redundant in the presence of
difference by $r \cap s = r - (r-s)$.  It appears that
intersection is no longer redundant in $\dat$; simulating
it using difference would increase the difference degree by two
times the number of nested applications of intersection.  It
remains open whether satisfiability of $\dat$ expressions not
using intersection is still undecidable.
\end{remark}

\section{Reduction from context-free grammar universality}
\label{secproof}

Consider a context-free grammar $G = (\Sigma,V,S,P)$ with set of
terminals $\Sigma$, set of nonterminals $V$, start symbol $S$,
and set of productions $P$.  Then $G$ is called \emph{universal}
if $L(G)$, the language generated by $G$, equals $\Sigma^*$.
Universality of context-free grammars is a well-known undecidable
problem \cite{hu}. We will reduce the complementary problem,
nonuniversality, to finite satisfiability of $\dat$ expressions.
The reduction will be based on a variation of the idea behind
Example~\ref{infinity}.

For technical reasons, we consider only grammars without
empty productions, and redefine universality to mean that all
nonempty strings over $\Sigma$ belong to $L(G)$.  Clearly,
this notion of universality is still undecidable.

For any grammar $G$ as above we construct a vocabulary $\Gamma_G$
and a $\dat$-expression $e_G$ over $\Gamma_G$ as follows.
Choose three symbols
$\alpha$, $\omega$ and $X$ not in $\Sigma \cup V$, and define $\Gamma_G =
\Sigma \cup V \cup \{\alpha,\omega,X\}$.  We define: $$ e_G = \varphi_0 -
(\varphi_1 \cup
\varphi_2 \cup
\varphi_3 \cup
\varphi_4 \cup
\varphi_5 \cup
\varphi_6 \cup
\varphi_7), $$
where the subexpressions $\varphi_i$ are defined as follows.  We
use $\Sigma$ as a shorthand for $\bigcup_{b \in \Sigma} b$.
\begin{align*}
\varphi_0 & = \alpha \Sigma \omega \\
\varphi_1 & = \alpha \Sigma (\omega - \alpha) \\
\varphi_2 & = \alpha(\Sigma\alpha - \alpha) \\
\varphi_3 & = \bigcup_{Z_0 \to Z_1 \dots Z_n \in P} \alpha (Z_1
\cdots Z_n-Z_0) \alpha \\
\varphi_4 & = (\alpha - \alpha \Sigma) S \omega \\
\varphi_5 & = \alpha (\Sigma - X) \alpha \\
\varphi_6 & = \alpha (XX - X) \alpha \\
\varphi_7 & = \alpha (X\Sigma \cap \Sigma) \alpha
\end{align*}

\begin{proposition}
$G$ is nonuniversal if and only if $e_G$ is finitely satisfiable.
\end{proposition}
\begin{proof}
The proof idea is an elaboration of the idea behind
Example~\ref{infinity}.
For the only-if direction, assume there exists a
nonempty word $b_1 \dots b_n$ not in $L(G)$.  We must show that
$e_G$ is finitely satisfiable.  Thereto
we construct the following structure $I$ over $\Gamma_G$:
\begin{align*}
\alpha^I & = \{(0,i) \mid i\in \{1,\dots,n\}\} \cup
\{(i,\infty) \mid i \in \{1,\dots,n+1\}\} \\
\omega^I & = \{(n+1,\infty)\} \\
b^I & = \{(i,i+1) \mid i \in \{1,\dots,n\} \ \& \  b_i=b\} \qquad
\text{for $b \in \Sigma$} \\
X^I & = \{(i,j) \mid i,j \in \{1,\dots,n\} \ \& \ i<j\} \\
Y^I & = \{(i,j) \mid i,j \in \{1,\dots,n\} \ \& \ i<j \ \& \
b_i \dots b_{j-1} \in L(G,Y)\} \qquad \text{for $Y \in V$}
\end{align*}
Here, $L(G,Y)$ is the set of words that can be generated from the
nonterminal $Y$.

We claim that $(0,\infty) \in e_G(I)$.  That 
$(0,\infty) \in \alpha \Sigma \alpha(I)$, and that
$$ (0,\infty) \notin (\varphi_1 \cup \varphi_2 \cup \varphi_3
\cup \varphi_5 \cup \varphi_6 \cup \varphi_7)(I), $$ can be
straightforwardly verified.  To see that 
$(0,\infty) \notin \varphi_4(I)$, assume the contrary.  Then
there exist edges
$0 \edge \alpha i \edge S j \edge \omega \infty$ in $I$ so that
$(0,i) \in (\alpha-\alpha\Sigma)(I)$.  This is only possible for
$i=1$ and $j=n+1$.  But then there is no edge $i \edge S j$ in
$I$ because $b_1\dots b_n \notin L(G)$.  Hence we have a
contradiction.

For the converse direction, assume that $G$ is universal.  We
show that $e_G$ is not finitely satisfiable.  It will be
convenient to assume that $G$ is in Chomsky normal form \cite{hu}, so that
every production is of one of the two forms $Z_0 \to Z_1 Z_2$ or
$Z_0 \to b$, with $Z_1,Z_2 \in V$ and $b \in \Sigma$.

Suppose, for the
sake of contradiction, that some pair $(x,y)$ belongs to the
result of $e_G$ evaluated in some finite structure $I$.  To avoid
clutter, in what follows we omit explicit references to $I$.
Since $(x,y) \in \varphi_0$, there exist edges $x \edge \alpha
u_2 \edge {b_1} u_1 \edge \omega y$ for some $b_1 \in \Sigma$.
Since $(x,y) \notin
\varphi_1$, we have also $u_1 \edge \alpha y$, and since $(x,y)
\notin \varphi_2$, we have $u_2 \edge \alpha y$ as well.
Since $(x,y) \notin \varphi_3$, we have
$u_2 \edge Y u_1$ for every production $Y \to b_1$ in $P$.

The above construction of $u_1$ and $u_2$ forms the basis for the
inductive construction of an infinite sequence $u_1,u_2,\dots$ so
that the following properties are satisfied for every
natural number $n \geq 2$:
\begin{enumerate}
\item \label{one}
$u_1 \edge \omega y$;
\item \label{two}
$x \edge \alpha u_i$ for each $2 \leq i\leq n$, and $u_i \edge \alpha y$ for
each $1 \leq i \leq n$;
\item \label{three}
for each $1 \leq i \leq n-1$ there is an edge
$u_{i+1} \edge {b_i} u_i$ with $b_i \in \Sigma$;
\item \label{four}
for every $Y \in V$ and every $n \geq j > i \geq 1$ such that
$b_{j-1} \dots b_i \in L(G,Y)$, 
there is an edge $u_j \edge Y u_i$.
\end{enumerate}

Specifically, for any $m \geq 2$, assume we already have defined
$u_1,\dots,u_m$; we then define $u_{m+1}$ as follows.
Since $G$ is universal, $b_{m-1}\dots b_1 \in L(G)$.
Hence, by property~(\ref{four}) above, $u_m \edge S u_1$.  Since
$(x,y) \notin \varphi_4$, there must exist an element $u$ with
edges $x \edge \alpha u \edge {b_m} u_m$ for some $b_m \in \Sigma$.
We set $u_{m+1} := u$ and check that the above properties are
still satisfied.

For property~(\ref{one}) nothing has changed.  For
property~(\ref{two}) we have $x \edge \alpha u_{m+1}$ given, and
$u_{m+1} \edge \alpha y$ follows from $(x,y) \notin \varphi_2$.
For property~(\ref{three}), we have $u_{m+1} \edge {b_m} u_m$
given.  For property~\ref{four}, we verify this by induction on
the length of the string $b_{j-1} \dots b_i$.  If $j=i+1$, the
production $Y \to b_i$ belongs to $P$ and we have $u_j \edge Y
u_i$ by $(x,y) \notin \varphi_3$.  If $j>i+1$, consider a
derivation tree of $b_{j-1} \dots b_i$ from $Y$, and let $Y \to
Z_1 Z_2$ be the production used at the root of the derivation
tree.  Then there exists $k$ strictly between $j$ and $i$ so that
$b_{j-1} \dots b_k \in L(G,Z_1)$ and $b_{k-1}\dots b_i \in
L(G,Z_2)$.  By induction we have edges $u_j \edge {Z_1} u_k \edge
{Z_1} u_i$, which implies $u_j \edge Y u_i$ by 
$(x,y) \notin \varphi_3$.

Now since $(x,y) \notin \varphi_5$, we have $u_{i+1} \edge X u_i$ for
each $i \geq 1$.   Then since $(x,y) \notin \varphi_6$, we have
$u_j \edge X u_i$ for all $j > i \geq 1$.  Since the structure is
finite, there must exist $1 \leq i < j$ so that $u_i=u_j$.  Hence
we have a self-loop $u_j \edge X u_j$, implying $(x,y) \in
\varphi_7$ which is in contradiction with $(x,y) \in e_G$.
\end{proof}

\section{Reduction to a single relation name} \label{secone}

In this section we establish our main theorem:

\begin{theorem}
The finite satisfiability problem for $\dato$ is undecidable.
\end{theorem}

The result of the previous Section already implies the
undecidability of the finite satisfiability problem for $\dat$.
Hence, to prove the above Theorem, it suffices to translate any given
expression $e$ over any given vocabulary $\Gamma$ to an
expression $e'$ over a single relation name, so that $\deg e =
\deg {e'}$ and $e$ is satisfiable if and only if
$e'$ is.

We will do this in two steps.  In a first step, we will reduce to
two relation names; in the second step we reduce further from two
to one.

Let $\Gamma=\{a_1,\dots,a_k\}$ ordered in an
arbitrary manner and let $b$ and $c$ be two symbols not in $\Gamma$.  We
define $e'$ as the expression obtained from $e$ by replacing
every occurrence of $a_i$, for $i=1,\dots,k$, by $b(c \cap
c^{i+1})b$, where $c^j$ denotes the composition $c\cdots c$ ($j$
times).

\begin{proposition}
$e$ is finitely satisfiable if and only if $e'$ is.
\end{proposition}
\begin{proof}
For the if-direction, we convert any structure $J$ over $\{b,c\}$ 
to a structure $K$ over $\Gamma$ as follows:  for each $a_i \in
\Gamma$, we set $a_i^K = b(c\cap
c^{i+1})b(J)$.  It is now readily verified by structural
induction that $e'(J)=e(K)$ for every expression $e$.  In
particular, if $e'(J)$ is nonempty, then so is $e(K)$.

For the only-if direction, we convert any structure $K$ over
$\Gamma$ to a structure $J$ over $\{b,c\}$ as follows.
Recall \cite{ahv_book} that the \emph{active domain} of $K$,
denoted by $\adom K$, equals the set of all elements that appear
as first or second component of a pair in a relation of $K$.  Now for each
$i=1,\dots,k$ and each $(x,y) \in a_i^K$, choose a set
$\{u^{x,y,i}_1, \dots, u^{x,y,i}_{i+2}\}$ of $i+2$ distinct
elements.  All these sets must be pairwise disjoint and disjoint
from $\adom K$.  Then $b^J$
consists of all edges $x \to u^{x,y,i}_1$ and $u^{x,y,i}_{i+2}
\to y$ for every $i=1,\dots,k$ and every $(x,y) \in a_i^K$.
Moreover $c^J$ consists of all edges $$ u^{x,y,i}_1 \to \cdots
\to u^{x,y,i}_{i+2} \ \text{and} \ u^{x,y,i}_1 \to
u^{x,y,i}_{i+2} $$ for every $i=1,\dots,k$ and every $(x,y) \in
a_i^K$.

For every expression $e$ we now again claim that $e(K)=e'(J)$.  
We can prove this again by induction on the structure of $e$.
The only potential difficulty is present in the basis of
the induction, where $e$ is a relation name $a_i \in \Gamma$.
The inclusion $e(K) \subseteq e'(J)$ holds by construction.
For the converse inclusion, assume $(u,v) \in b(c \cap c^{i+1})b(J)$.  
Then there exist edges $u \edge b z_1 \edge c z_2 \edge b v$ such that
$(z_1,z_2) \in c^{i+1}(J)$.  Due to the edge $u \edge b z_1$,
there are only two possibilities for $u$:
\begin{itemize}
\item
$u$ equals $u^{x,y,j}_{j+2}$,
for some $x$, $y$ and $j$ such that $(x,y) \in a_j^K$.
Then $z_1$ must be $y$.  However, by $z_1 \edge c z_2$, this is
impossible, since there is no $c$-edge leaving $y$.
\item
$u$ equals $x$, for some $y$ and $j$ such that $(x,y) \in a_j^K$.
Then $z_1$ is $u^{x,y,j}_1$ and there are two possibilities for $z_2$:
\begin{enumerate}
\item
$z_2$ is $u^{x,y,j}_2$.  By $z_2 \edge b v$ this is impossible,
since there is no $b$-edge leaving $u^{x,y,j}_2$.
\item
$z_2$ is $u^{x,y,j}_{j+2}$, so $v$ is $y$.  Since
$(z_1,z_2) \in c^{i+1}(J)$, and the only chain of $c$-edges from
$u^{x,y,j}_1$ to $u^{x,y,j}_{j+2}$ is the chain
$ u^{x,y,j}_1 \to \cdots
\to u^{x,y,j}_{j+2}$, we must have $j=i$.  Hence, we obtain that
$(u,v)=(x,y) \in a_i^K$ as desired.
\end{enumerate}
\end{itemize}
\end{proof}

For the reduction to a single relation name, consider any
expression $e$ over the vocabulary $\{b,c\}$ with two relation
names, and let $a$ be a third symbol.  We define the expression
$\hat e$ over the vocabulary $\{a\}$ as the expression obtained
from $e$ by replacing every occurrence of $b$ by $a(a \cap a^2)a$
and every occurrence of $c$ by $a(a \cap a^3)a$.  Again we show:

\begin{proposition}
$e$ is finitely satisfiable if and only if $\hat e$ is.
\end{proposition}
\begin{proof}
For the if-direction, we convert any structure $J$ over $\{a\}$ 
to a structure $I$ over $\{b,c\}$ as follows:
$b^I = a(a \cap a^2)a(J)$ and $c^I = a(a \cap a^3)a(J)$.
It is now readily verified by structural
induction that $\hat e(J)=e(I)$ for every expression $e$ over
$\{b,c\}$.  In particular, if $\hat e(J)$ is nonempty, then so is $e(I)$.

For the only-if direction, we convert any structure $I$ over
$\{b,c\}$ to a structure $J$ over $\{a\}$ as follows.
For every edge $x \edge b y$ in $I$ we choose a set
$\{u^{x,y,b}_1,\dots, u^{x,y,b}_3\}$ of three distinct elements;
for every edge $x \edge c y$ in $I$ we choose a set
$\{u^{x,y,c}_1,\dots, u^{x,y,c}_4\}$ of four distinct elements.
All these sets must be pairwise disjoint and disjoint from $\adom
I$. We now define $a^J$ to consist of all edges
$$ x \to
u^{x,y,b}_1 \to
u^{x,y,b}_2 \to
u^{x,y,b}_3 \to y \ \text{and} \ 
u^{x,y,b}_1 \to u^{x,y,b}_3 $$
for every edge $x \edge b y$ in $I$, plus all edges
$$ x \to
u^{x,y,b}_1 \to
u^{x,y,b}_2 \to
u^{x,y,b}_3 \to
u^{x,y,b}_4 \to y \ \text{and} \ 
u^{x,y,b}_1 \to u^{x,y,b}_4 $$
for every edge $x \edge c y$ in $I$.

We now make Claim~B and Claim~C.
\begin{description}
\item[Claim~B:] $b^I = a(a\cap a^2)a(J)$.  The inclusion from left
to right holds by construction.  For the inclusion from right to
left, let $(u,v) \in a(a \cap a^2)a(J)$.  Then there exist edges
$u \to z_1 \to z_2 \to v$ in $J$ with $(z_1,z_2) \in (a \cap a^2)(J)$.
An obvious possibility is that $z_1 = u^{x,y,b}_1$ and $z_2 =
u^{x,y,b}_3$ for some $(x,y) \in b^I$.  Then $u$ must equal $x$
and $v$ must equal $y$ so $(u,v) = (x,y) \in b^I$ as desired.
Let us now verify that there are no other possibilities for $z_1$
and $z_2$.  Thereto we list all  other possibilities for a pair $(z_1,z_2)
\in a^2(J)$:
\begin{itemize}
\item
$(u^{x,y,b}_1,y)$ with $x$ and $y$ as above;
\item
$(u^{x,y,b}_2,y)$;
\item
$(x,u^{x,y,b}_2)$;
\item
$(x,u^{x,y,b}_3)$;
\item
$(u^{x,y,b}_3,u^{y,z,r}_1)$, with $r=b$ or $c$, for some $z$ such that
$(y,z) \in r^I$;
\item
$(u^{x',y',c}_1,y')$ for some $(x',y') \in c^I$;
\item
$(u^{x',y',c}_1,u^{x',y',c}_3)$;
\item
$(u^{x',y',c}_2,u^{x',y',c}_4)$;
\item
$(u^{x',y',c}_3,y')$;
\item
$(u^{x',y',c}_4,u^{y',z',r}_1)$, with $r=b$ or $c$, for some $z'$ such that
$(y',z') \in r^I$;
\item
$(x',u^{x',y',c}_2)$;
\item
$(x',u^{x',y',c}_4)$. 
\end{itemize}
In all these cases, there is no edge $z_1 \to z_2$ in $J$, so
that $(z_1,z_2) \notin (a\cap a^2)(J)$.

\item[Claim C:]
$c^I = a(a\cap a^3)a(J)$.  The inclusion from left
to right holds by construction.  For the inclusion from right to
left, let $(u,v) \in a(a \cap a^3)a(J)$.  Then there exist edges
$u \to z_1 \to z_2 \to v$ in $J$ with $(z_1,z_2) \in (a \cap a^3)(J)$.
The obvious possibility is that $z_1 = u^{x,y,c}_1$ and $z_2 =
u^{x,y,c}_4$ for some $(x,y) \in c^I$.  Then $u$ must equal $x$
and $v$ must equal $y$ so $(u,v) = (x,y) \in c^I$ as desired.
We now verify that there are no other possibilities for $z_1$
and $z_2$.  Thereto we list all other possibilities for a pair $(z_1,z_2)
\in a^3(J)$:
\begin{itemize}
\item
$(u^{x,y,c}_1,u^{y,z,r}_1)$,
with $x$ and $y$ as above,
$r=b$ or $c$, and some $z$ such that $(y,z) \in r^I$;
\item
$(u^{x,y,c}_2,y)$;
\item
$(u^{x,y,c}_3,u^{y,z,r}_1)$;
\item
$(u^{x,y,c}_4,u^{y,z,r}_2)$;
\item
$(u^{x,y,c}_4,u^{y,z,r}_3)$ if $r=b$;
\item
$(u^{x,y,c}_4,u^{y,z,r}_4)$ if $r=c$;
\item
$(x,u^{x,y,c}_3)$;
\item
$(x,y)$;
\item
$(u^{x',y',b}_1,y')$ for some $(x',y') \in b^I$;
\item
$(u^{x',y',b}_1,u^{y',z',r'}_1)$,
with $r'=b$ or $c$, for some $z'$ such that $(y',z') \in r'^I$;
\item
$(u^{x',y',b}_2,u^{y',z',r'}_1)$;
\item
$(u^{x',y',b}_3,u^{y',z',r'}_2)$;
\item
$(u^{x',y',b}_3,u^{y',z',r'}_3)$ if $r'=b$;
\item
$(u^{x',y',b}_3,u^{y',z',r'}_4)$ if $r'=c$;
\item
$(x',u^{x',y',b}_3)$;
\item
$(x',y')$.
\end{itemize}
In all these cases, there is no edge $z_1 \to z_2$ in $J$, so
that $(z_1,z_2) \notin (a\cap a^3)(J)$.
\end{description}

From Claims B and C it now follows readily by structural
induction that $ e(I)=\hat e(J) $ for every expression $e$ over
$\{b,c\}$.  In particular, if $e(I)$ is nonempty, then so is
$\hat e(J)$.
\end{proof}

\section{Conclusion} \label{seconcl}

In $\dat$-expressions, applications of the set difference
operation can be nested at most once.  It is thus natural to
wonder what happens in the fragment where set difference cannot
be nested at all.  In a companion paper, we consider the fragment
of the full Tarski Algebra (TA), with general complementation,
defined by the restriction that
complement can only be applied to expressions that do not already
contain an application of complement.  It turns out that finite
satisfiability for TA-expressions without nested complement is decidable and
even belongs to NP.

As already mentioned in Remark~\ref{remarkintersection},
it remains open whether
satisfiability for $\dat$-expressions without the intersection
operation is decidable.
As already mentioned in the Introduction,
the decidability of unrestricted satisfiability for DA remains
open as well.

\section*{Acknowledgment}

We are indebted to Stijn Vansummeren for a number of inspiring
discussions on the topic of this paper.

\bibliographystyle{alpha}
\bibliography{database}

\newcommand{\etalchar}[1]{$^{#1}$}
\begin{thebibliography}{BCM{\etalchar{+}}03}

\bibitem[AGN97]{andreka_memoir}
H.~Andr\'eka, S.~Givant, and I.~N\'emeti.
\newblock {\em Decision problems for equational theories of relational
  algebras}, volume 126 of {\em Memoirs}.
\newblock AMS, 1997.

\bibitem[AHV95]{ahv_book}
S.~Abiteboul, R.~Hull, and V.~Vianu.
\newblock {\em Foundations of Databases}.
\newblock Addison-Wesley, 1995.

\bibitem[BCM{\etalchar{+}}03]{dlhandbook}
F.~Baader, D.~Calvanese, D.~McGuiness, D.~Nardi, and P.~Patel-Schneider,
  editors.
\newblock {\em The Description Logic Handbook}.
\newblock Cambridge University Press, 2003.

\bibitem[BGG97]{bgg_decision}
E.~B\"orger, E.~Gr\"adel, and Y.~Gurevich.
\newblock {\em The Classical Decision Problem}.
\newblock Springer, 1997.

\bibitem[DG06]{dg_nfbr}
D.J. Dougherty and C.~Gutierrez.
\newblock Normal forms for binary relations.
\newblock {\em Theoretical Computer Science}, 360(1--3):228--246, 2006.

\bibitem[FGL{\etalchar{+}}11]{rafragments}
G.H.L. Fletcher, M.~Gyssens, D.~Leinders, J.~Van~den Bussche, D.~Van~Gucht,
  S.~Vansummeren, and Y.~Wu.
\newblock Relative expressive power of navigational querying on graphs.
\newblock In {\em Proceedings 14th International Conference on Database
  Theory}, 2011.

\bibitem[Fig12]{figueira_downward}
D.~Figueira.
\newblock Decidability of downward {XPath}.
\newblock {\em ACM Transactions on Computational Logic}, 13(4):article 34,
  2012.

\bibitem[HKT00]{dynamiclogicbook}
D.~Harel, D.~Kozen, and J.~Tiuryn.
\newblock {\em Dynamic Logic}.
\newblock MIT Press, 2000.

\bibitem[HU79]{hu}
J.E. Hopcroft and J.D. Ullman.
\newblock {\em Introduction to Automata Theory, Languages, and Computation}.
\newblock Addison-Wesley, 1979.

\bibitem[KRV14]{vrgoc_containment}
E.V. Kostylev, J.L. Reutter, and D.~Vrgo\v{c}.
\newblock Containment of data graph queries.
\newblock In {\em Proceedings 17th International Conference on Database
  Theory}. ACM, 2014.

\bibitem[LMV13]{gxpath}
L.~Libkin, W.~Martens, and D.~Vrgo\v{c}.
\newblock Quering graph databases with {XPath}.
\newblock In {\em Proceedings 16th International Conference on Database
  Theory}. ACM, 2013.

\bibitem[LW05]{lutz_atomneg}
C.~Lutz and D.~Walther.
\newblock {PDL} with negation of atomic programs.
\newblock {\em Journal of Applied Non-Classical Logics}, 15(2):189--213, 2005.

\bibitem[Mad91]{maddux_originra}
R.D. Maddux.
\newblock The origin of relation algebras in the development and axiomatization
  of the calculus of relations.
\newblock {\em Studia Logica}, 50(3/4):421--455, 1991.

\bibitem[MV97]{marxvenema_multi}
M.~Marx and Y.~Venema.
\newblock {\em Multi-Dimensional Modal Logic}.
\newblock Springer, 1997.

\bibitem[Pra92]{pratt_relcalc}
V.~Pratt.
\newblock Origins of the calculus of binary relations.
\newblock In {\em Proceedings 7th Annual IEEE Symposium on Logic in Computer
  Science}, pages 248--254, 1992.

\bibitem[Sch79]{schoenfeld_undecfinra}
W.~Sch\"onfeld.
\newblock An undecidability result for relation algebras.
\newblock {\em Journal of Symbolic Logic}, 44(1):111--115, 1979.

\bibitem[SSVG93]{good_tarski}
V.M. Sarathy, L.V. Saxton, and D.~Van~Gucht.
\newblock Algebraic foundation and optimization for object based query
  languages.
\newblock In {\em Proceedings 9th International Conference on Data
  Engineering}, pages 81--90. IEEE Computer Society, 1993.

\bibitem[Tar41]{tarski_relcalc}
A.~Tarski.
\newblock On the calculus of relations.
\newblock {\em Journal of Symbolic Logic}, 6:73--89, 1941.

\bibitem[tCM07]{catemarx_sigmodrecord}
B.~ten Cate and M.~Marx.
\newblock Navigational {XPath}: Calculus and algebra.
\newblock {\em SIGMOD Record}, 36(2):19--26, 2007.

\bibitem[TG87]{tarskigivant}
A.~Tarski and S.~Givant.
\newblock {\em A Formalization of Set Theory Without Variables}, volume~41 of
  {\em AMS Colloquium Publications}.
\newblock American Mathematical Society, 1987.

\bibitem[VTL82]{valdes_seriesparallel}
J.~Valdes, R.E. Tarjan, and E.L. Lawler.
\newblock The recognition of series parallel digraphs.
\newblock {\em SIAM Journal on Computing}, 11:298--313, 1982.

\bibitem[Woo12]{wood_survey}
P.~Wood.
\newblock Query languages for graph databases.
\newblock {\em SIGMOD Record}, 41(1):50--60, March 2012.

\end{thebibliography}

\end{document}